\documentclass[11pt]{article}
\usepackage{fullpage,amsfonts,amssymb,amsmath,amsthm,float,graphicx,todonotes}
\usepackage[breaklinks]{hyperref}

\newtheorem{theorem}{Theorem}
\newtheorem{lemma}[theorem]{Lemma}
\newtheorem{corollary}[theorem]{Corollary}
\newtheorem{lemmaPrime}{Lemma}
\theoremstyle{definition}
\newtheorem{definition}[theorem]{Definition}

\def\A{{\cal A}}
\newcommand{\ket}[1]{\mathinner{\left \lvert #1 \right\rangle}}
\newcommand{\braket}[2]{\mathinner{ \left\langle #1 \right. \left \lvert #2 \right\rangle}}
\newcommand\lr[1]{\left( #1 \right)}
\newcommand\lra[1]{\left\langle  #1 \right\rangle}
\newcommand\lrv[1]{\left|  #1 \right|}
\newcommand\lrk[1]{\left[  #1 \right]}
\newcommand\lrb[1]{\left\lbrace #1 \right\rbrace}
\newcommand\nrm[1]{\left\Vert #1 \right\Vert}

\newcommand{\tdots}{\,..\,}
\newcommand{\eps}{\epsilon}

\DeclareMathOperator{\bmdeg}{bmdeg}

\newcommand{\bmdega}{\widetilde \bmdeg}
\newcommand{\dega}{\widetilde \deg}
\DeclareMathOperator*{\Q}{Q}

\newcommand{\BCF}{\lrb{-1,1}}

\newcommand{\mbb}[1]{\mathbb{#1}}
\newcommand{\mc}[1]{\mathcal{#1}}

\def\lbra{\langle}
\newcommand{\bra}[1]{\langle #1|}

\newcommand{\nonzVec}[1]{\mbb R_{+}^{#1}}
\newcommand{\obtain}{\longrightarrow }
\newcommand{\lmax}{\lambda_{\max}}

\begin{document}

\title{Polynomials, Quantum Query Complexity, and Grothendieck's Inequality}
\date{}
\author{Scott Aaronson$^1$
	\and Andris Ambainis$^2$
	\and J\=anis Iraids$^2$
	\and Martins Kokainis$^2$
	\and Juris Smotrovs$^2$
}
\maketitle

\begin{abstract}
We show an equivalence between 1-query quantum algorithms and representations by degree-2 polynomials.
Namely, a partial Boolean function $f$ is computable by a 1-query quantum algorithm with error bounded by $\epsilon<1/2$
iff $f$ can be approximated by a degree-2 polynomial with error bounded by $\epsilon'<1/2$.
This result holds for two different notions of approximation by a polynomial: the standard definition of Nisan and Szegedy \cite{NS} and
the approximation by block-multilinear polynomials recently introduced by Aaronson and Ambainis \cite{AA}.

We also show two results for polynomials of higher degree. First, there is a total Boolean function which requires $\tilde{\Omega}(n)$ quantum queries
but can be represented by a block-multilinear polynomial of degree $\tilde{O}(\sqrt{n})$. Thus, in the general case 
(for an arbitrary number of queries), block-multilinear polynomials are not equivalent to quantum algorithms.

Second, for any constant degree $k$, the two notions of approximation by a polynomial (the standard and the block-multilinear) are equivalent. As a consequence, we solve an open problem from \cite{AA}, showing that
one can estimate the value of any bounded degree-$k$ polynomial $p:\{0, 1\}^n \rightarrow [-1, 1]$ with $O(n^{1-\frac{1}{2k}})$ queries.
\end{abstract}

\footnotetext[1]{Computer Science and Artificial Intelligence Laboratory, MIT. Supported by an Alan T. Waterman Award from the National Science Foundation, under grant no. 1249349. E-mail: {\tt aaronson@csail.mit.edu}.}
\footnotetext[2]{Faculty of Computing, University of Latvia.
	Supported by the European Commission FET-Proactive project QALGO, ERC Advanced Grant MQC and Latvian State Research programme NexIT project No.1. Emails:{\tt ambainis@lu.lv, janis.iraids@gmail.com,juris.smotrovs@sets.lv,
martins.kokainis@lu.lv}.}


\setcounter{page}{0}
\thispagestyle{empty}
\newpage

\section{Introduction}


Many of the known quantum algorithms can be studied in the query model where one measures the complexity of an algorithm by the number of queries to the input that it makes. In particular, this model encompasses Grover's search \cite{Grover}, the quantum part of Shor's factoring algorithm (period-finding) \cite{Shor}, their generalizations and many of the more recent quantum algorithms such as element distinctness \cite{A04} and NAND tree evaluation \cite{FGG,AC+,Reichardt}.

For proving lower bounds on quantum query algorithms, one often uses a connection to polynomials \cite{Beals}. \ After $k$ queries to an input $x_1, \ldots, x_N$,  the amplitudes of the algorithm's quantum state are polynomials of degree at most $k$ in $x_1, \ldots, x_N$. \ Therefore, one can prove that there is no quantum algorithm using fewer than $k$ queries by showing the non-existence of a polynomial with certain properties.

For example, one can use this approach to show that any quantum algorithm for Grover's search algorithm requires 
$\Omega(\sqrt{N})$ queries \cite{Beals} or to show an optimal quantum lower bound for finding collisions \cite{AS}.
In some cases, the lower bounds obtained by polynomials method are tight, either exactly (for example, for computing the parity
of $N$ input bits $x_1, \ldots, x_N$ \cite{Beals}) or up to a constant factor (Grover's search and many other examples).
In other cases, the number of queries to compute a function $f(x_1, \ldots, x_N)$ is asymptotically larger than the lower bound
which follows from polynomials \cite{A03,ABK}.

In this paper, we discover the first case where we can go in the opposite direction: from a polynomial to a bounded-error quantum algorithm\footnote{In unbounded-error settings, equivalences between quantum algorithms and polynomials were previously shown by de Wolf \cite{deWolf} and by Montanaro et al.\ \cite{MNR}.}. 
That is, polynomials with certain properties and quantum algorithms are equivalent! 

In more detail, we consider computing partial Boolean functions $f(x_1, \ldots, x_n)$ and show that 
the existence of a quantum algorithm that computes $f$ with 1 query is equivalent to the existence of a degree 2 polynomial that approximates
$f$. \ This result holds for two different notions of approximation by a polynomial: the standard one in \cite{NS} and
the approximation by block-multilinear polynomials introduced in \cite{AA}.

The methods that we use are quite interesting. To transform a polynomial into a quantum algorithm, we first transform it into the block-multilinear form of \cite{AA} and then use a variant of Grothendieck's inequality for relating two matrix norms \cite{Pisier}. One of the two norms corresponds to the constraints on the block-multilinear polynomials while the other norm corresponds to algorithm's transformations being unitary. While Grothendieck's inequality has been used in the context of quantum non-locality (e.g. in \cite{AGT}), this appears to be its first use in the context of quantum algorithms.

We then show two results for polynomials of larger degree:
\begin{itemize}
\item
similarly to general polynomials, block-multilinear polynomials are not equivalent to quantum algorithms in the general case: one of {\em cheat-sheet} functions of \cite{ABK} requires $\tilde{\Omega}(n)$ quantum queries but can be described by a block-multilinear polynomial of degree $\tilde{O}(\sqrt{n})$;
\item
for representations by polynomials of degree $d=O(1)$, a partial function $f$ can be represented by a general polynomial of degree $d$ if and only if it can be represented by a block-multilinear polynomial of degree $d$.
\end{itemize}
We note that the first result does not exclude an equivalence between quantum algorithms and polynomials for a small number of queries that is larger than 1. \ For example, 2-query quantum algorithms could be equivalent to polynomials of degree 4. The second result shows that, to prove such an equivalence, it suffices to give a transformation from block-multilinear polynomials to quantum algorithms. 

Another consequence of the second result is that, if we have a general polynomial $f(x_1, \ldots, x_n)$ which is bounded (i.e., $|f|\leq 1$ for all $x_1, \ldots, x_n\in\{0, 1\}$), the value of this polynomial can be estimated with $O(n^{1-1/2d})$ queries about values of $x_1, \ldots, x_n$. 
This resolves an open problem from \cite{AA} and is shown by transforming $f$ into a 
block-multilinear form and then using the sampling algorithm of \cite{AA} for block-multilinear polynomials.
\section{Preliminaries}
\subsection{Notation}
By $  [a\tdots b] $, with $ a,b $ being integers, $ a\leq b  $, we denote the set $ \lrb{a,a+1,a+2, \ldots, b} $. When $ a=1 $, notation $ [a \tdots b] $ is simplified to $ [b] $.

For a vector  $ x $, let    $ \nrm x_p $ stand for the $ p $-norm; when $ p=2 $, this is the Euclidean norm and the notation is simplified to $ \nrm x  $. For a matrix $ A $, by  $ \nrm{A}_{p \to q} $ we denote
\[
\nrm{A}_{p \to q} = \sup_{x : \nrm x_p  \neq 0}  \frac{\nrm {Ax}_q}{\nrm x_p }= \max_{x : \nrm x_p  =1}   \nrm {Ax}_q= \max_{x : \nrm x_p  \leq 1}   \nrm {Ax}_q.
\]
By $ \nrm A  $ we understand the usual operator norm $ \nrm{A}_{2 \to 2} $.

$ D_x $ stands for the diagonal matrix with components of $ x $ on its diagonal.

By $ K $ we denote the (real) Grothendieck's constant which is defined as the smallest number with the following 
property: if $A=(a_{ij})$ is such that $\sum_{i, j} a_{ij} x_i y_j\leq 1$ for any choice of $x_i, y_j\in\{-1, 1\}$,
then $\sum_{i, j} a_{ij} (u_i, v_j)\leq K$ for any choice of vectors (with real components) $u_i, v_j$ with $\|u_i\|=1$ and $\|v_j\|=1$ for all $i, j$. It is known \cite{Pisier,Braverman} that
\[ \frac{\pi}{2} \leq K < \frac{\pi}{2\ln(1 + \sqrt 2)} .\]

\subsection{Quantum query complexity and polynomial degree}

We consider computing partial Boolean functions $f(x_1, \ldots, x_n):X \rightarrow \{0, 1\}$ (for some $X\subseteq \{0, 1\}^n$) in the standard quantum query model. For technical convenience, we relabel the values of input variables $x_i$ from $\{0, 1\}$ to $\{-1, 1\}$. Then, a partial Boolean function $f$ maps a set $X\subseteq \{-1, 1\}^n$ to $\{0, 1\}$.

Let $\Q_{\epsilon}(f)$ be the minimum number of queries in a quantum algorithm computing
$f$ correctly with probability at least $1-\epsilon$, for every $x=(x_1, \ldots, x_n)$ for which $f(x)$ is defined.

\begin{definition}
\label{def:deg}
$\dega_{\epsilon}(f)$ is the minimum degree of a polynomial $p$ (in variables $x_1, \ldots, x_n$) such that
\begin{enumerate}
\item
$\lrv{p (x)  -f(x) } \leq  \epsilon $ for all $x\in\{-1, 1\}^n$ for which $f(x)$ is defined;
\item
$p(x)\in[0, 1]$ for all $x\in\{-1, 1\}^n$.
\end{enumerate}
$\deg(f)$ denotes $\dega_0(f)$.
\end{definition}

It is well known that $\Q_{\epsilon}(f)\geq \frac{1}{2} \dega_{\epsilon}(f)$ \cite{Beals}. We now consider a refinement of this result due to \cite{AA}.
We say that a polynomial $p$ of degree $k$ is \textit{block-multilinear} if its variables $x_{1} ,\ldots,x_{N}$ can be partitioned into $k$ blocks, $R_{1},\ldots,R_{k}$, so  that every monomial of $p$ contains exactly one variable from each block\footnote{In other words, a block-multilinear polynomial is just a multilinear form.
	We, however, use the word {\em block-multilinear}, to emphasize the difference from standard polynomial representations of Boolean functions which are multilinear
	but are not multilinear forms.}.   

\begin{lemma}[{\cite[Lemma 20]{AA}}] \label{blockbeals}
	Let $\mathcal{A}$ be a quantum algorithm that makes $t$ queries to a Boolean input $x\in\BCF  ^{n}$.  Then there	exists a degree-$2t$ block-multilinear polynomial $p:\mathbb{R}^{2t (n+1)} \rightarrow\mathbb{R}$, with $2t$ blocks of $n+1$ variables each, such that
	\begin{enumerate}
		\item[(i)] the probability that $\mathcal{A}$ outputs 1 for an input $x=(x_1,\ldots,x_n) \in \BCF^n$  equals $p(\tilde x, \ldots, \tilde x) $, where $ \tilde x:= (1,x_1,\ldots,x_n) $   (with $\tilde x$ repeated $2t$  times), and	
		\item[(ii)] $p (z) \in\left[  -1,1\right]  $  for all $z\in\BCF^{2t(n+1)}$.
	\end{enumerate}
\end{lemma}

The first variable in each block (which is set to 1 in the requirement (i)) corresponds to the possibility that the
algorithm is not asking any of the actual variables $x_1, \ldots, x_n$ in a given query. (Although the statement of Lemma 20 in \cite{AA} does not mention such variables explicitly, they are used in the proof of the Lemma.)

\begin{definition}
\label{def:bmdeg}
Let the \textit{block-multilinear approximate degree} of $f$, or $\bmdega_{\eps}(f)$, be the minimum degree of any block-multilinear polynomial $p:\mathbb{R}^{k(n+1)}\rightarrow\mathbb{R}$, with $k$  blocks of $n+1$ variables each, such that
\begin{enumerate}
	\item[(i)] $\lrv{p \lr{\tilde x,\ldots, \tilde x}  -f(x) } \leq  \eps$ for all $x\in \BCF^{n}$ for which $f(x)$ is defined, and
	\item[(ii)] $p \lr{ x_{1,0 }, x_{1,1},\ldots, x_{1,n}, x_{2,0}, \ldots,x_{k,n}} \in \lrk{-1,1}  $ for all $x_{1,0},\ldots,x_{k,n}\in \BCF^{k(n+1)}$.
\end{enumerate}
$\bmdeg(f)$ denotes $\bmdega_0(f)$.
\end{definition}
As a particular case, this definition includes block-multilinear polynomials $ p : \mbb R^{kn} \to \mbb R $ which satisfy
\[
\forall x \in \BCF^n \  \lrv{p(x,\ldots,x) - f(x) } \leq \eps  \quad \text{and} \quad  \forall z \in \BCF^{kn} p(z) \in [-1,1],
 \]
because we can view them as polynomials $ p : \mbb R^{k(n+1)} \to \mbb R $ in which each monomial containing a variable
$ x_{1,0} $, $ x_{2,0} $, \ldots,  or $ x_{k,0} $ has a coefficient zero.

We have $\dega_{\epsilon}(f) \leq \bmdega_{\epsilon}(f) \leq 2 \Q_{\epsilon}(f)$. The first of the two inequalities follows by
taking $q(x)=p(\tilde x, \ldots, \tilde x)$. If $p$ satisfies the requirements of Definition \ref{def:bmdeg}, then $q$ satisfies the requirements of Definition \ref{def:deg}.
The second inequality follows from Lemma \ref{blockbeals}.

\subsection{Equivalence between block-multilinear and general polynomials}
\label{sec:equiv}

The two types of polynomial representations ($\dega$ and $\bmdega$) are equivalent to one another, up to some loss in
the quality of approximation. This has been shown independently by us and  by O'Donnell and Zhao \cite{DZ}:

\begin{theorem}
	\label{thm:dz}
	Let $p(x_1, \ldots, x_n)$ be a polynomial of degree $d$ with $|p(x_1, \ldots, x_n)|\leq 1$ for any $x\in\BCF^n$. 
          Then there is a block-multilinear polynomial 
	$\tilde{p}: R^{(n+1)d} \rightarrow R$ such that
	\begin{enumerate}
		\item
		$\tilde{p}(\tilde{x}, \ldots, \tilde{x}) = p(x)$ for any $x\in\BCF^n$;
		\item
		$|\tilde{p}(y)| \leq C_d$ for any $y\in\BCF^{(n+1)d}$ with $C_d$ being a constant that depends on
		the degree $d$ only.
	\end{enumerate}
\end{theorem}

O'Donnell and Zhao \cite{DZ} show $C_d\leq (2e)^d$.   In Appendix \ref{sec:higher}
we show our version of this result with $C_2 \leq 3$ for $d=2$  and $C_d = O(3.5911...^d)$. 

The result of O'Donnell and Zhao is a special case of  the general theory of {\em decoupling} \cite{Kwapien,PG}
which proves much more general results. In contrast, our proof is specific to the problem above but, due to that, 
it is quite simple and gives better constants $C_d$.

As a consequence of this Theorem, we have

\begin{corollary}
	\label{cor:equiv}
	Let $\epsilon$ be such that $0\leq \epsilon < \frac{1}{2}$ and let 
	$\epsilon' = \frac{1}{2} - \frac{1}{C_d} (\frac{1}{2}-\epsilon)$. Then
	$\dega_{\epsilon} (f)  \leq d$ implies  $\bmdega_{\epsilon'} (f) \leq d$.
\end{corollary}

\begin{proof}
	We take the polynomial $q$ which approximates
	$f(x_1, \ldots, x_n)$ with error $\epsilon$ according to Definition \ref{def:deg} and apply Theorem \ref{thm:dz} to
	$p(x_1, \ldots, x_n)= q(x_1, \ldots, x_n) - \frac{1}{2} $. Then the polynomial $\frac{1}{2}+ \frac{1}{2C_d}\tilde{p}$ approximates
	$f$ in the sense of Definition \ref{def:bmdeg}. 
\end{proof}

\subsection{Block-multilinear polynomials of degree 2}
\label{sec:deg2}

Let
\begin{equation}
\label{eq:d2}
p(x_1,\ldots,x_n,y_1,\ldots,y_m) = \sum_{\substack{i \in [ n]\\ j \in [m]}} a_{ij} x_{i} y_j,
\end{equation}
be a block-multilinear polynomial of degree 2, with the variables in the first block labeled as $ x_1, \ldots, x_n $ and the variables in the second block labeled as $ y_1,\ldots, y_m $. 
We say that $p$ is {\em bounded} 
if $|p(x_1, \ldots, x_n, y_1, \ldots, y_m)|\leq 1$ for all $x_1, \ldots, y_m\in\{-1, 1\}$.
Then, we have 
\[
\max_{\substack{
			x \in \BCF^n\\
			y \in \BCF^m
		}}
		\lrv{ \sum_{\substack{i \in [ n]\\ j \in [m]}} a_{ij} x_{i} y_j} \leq 1.
 \]
Let $ A $ be the $ n\times m $ matrix with entries $ a_{ij} $, then
\[
p(x,y) = x^T A y   \quad \text{for all }x \in \mbb R^n ,  \ y \in \mbb R^m
 \]
and $p$ being bounded translates to 
the infinity-to-1 norm of $ A $ being at most 1, i.e., $ \nrm{A}_{\infty\to 1} \leq 1 $.

\section{Equivalence between polynomials of degree 2 and 1-query quantum algorithms }

Let $f$ be a partial Boolean function. In this section, we show that 
the following two statements are equivalent\footnote{The equivalence here involves some loss in
	the error $\epsilon$. 
	However, the bound $\epsilon$ on the error probability of the resulting quantum algorithm 
	only depends on the error of the polynomial approximation from which
	we started and does not increase with the number of variables $n$.}:
\begin{enumerate}
	\item[(a)]
	$\Q_{\epsilon}(f)\leq 1$ for some $\epsilon$ with $0\leq \epsilon < \frac{1}{2}$;
	\item[(b)]
	$ \bmdega_{\epsilon'}(f) \leq 2 $ for some $\epsilon'$ with $0\leq \epsilon' < \frac{1}{2}$;
\end{enumerate}

Given (a), Lemma \ref{blockbeals} implies that (b) holds with $\epsilon'=\epsilon$.
We now show that (b) implies (a) with $\epsilon=\frac{K+\epsilon'}{2(K+1)}$ where $K$ is Grothendieck's constant.

Because of results in Section \ref{sec:equiv}, $ \dega_{\epsilon''}(f) \leq 2 $ implies $\bmdega_{\epsilon'}(f)\leq 2$ for $\epsilon'=\frac{1+\epsilon''}{3}$.
Therefore we also get a similar equivalence between 
$\Q_{\epsilon}(f)\leq 1$ and $ \dega_{\epsilon''}(f) \leq 2 $.

\begin{theorem}
Let $f$ be a partial Boolean function.  If $ \bmdega_{\epsilon'}(f) \leq 2 $, then 
$\Q_{\epsilon}(f)\leq 1$ for $\epsilon=\frac{K+\epsilon'}{2(K+1)}$.  
\end{theorem}

\begin{proof}
We start with two technical lemmas.

\begin{lemma}\label{th:embedB}
	If a $ n \times  m$ complex matrix $ B $ satisfies $ \nrm B \leq C $, then there exists a unitary $U$ (on a  possibly larger space  with basis states $\ket{1}, \ldots, \ket{k}$ for some $k\geq \max(n, m)$) such that, for any unit vector $\ket{y}= \sum_{i=1}^{m} \alpha_i \ket{i}$, $U\ket{y} = \frac{B\ket{y}}{C} + \ket{\phi}$, with $\ket{\phi}$ consisting of basis states $\ket{i}$, $ i>n$ only.
\end{lemma}

\begin{proof}
Without the loss of generality, we can assume that $C=1$ (otherwise, we just replace the matrix $B$ by $\frac{B}{C}$).

Let $A=I-B^{\dagger} B$. Since $\|B\|\leq 1$, the eigenvalues of $B^{\dagger} B$ are at most 1 and, hence,
$A$ is positive semidefinite. Let $A=V^{\dagger} \Lambda V$ be the eigendecomposition of $A$, with $V$ being a
unitary matrix and $\Lambda$ a diagonal matrix. We take $W=\sqrt{\Lambda} V$. Then, $A=W^{\dagger} W$
and, if we take the block matrix $U=\left( \begin{array}{c} B  \\ W\end{array}\right)$, 
we get $U^{\dagger} U = B^{\dagger} B + W^{\dagger} W = I$.

Let $k \times m$ be the size of the matrix $U$. For any $i\in\{1, \ldots, m\}$, we have
$\bra{i} U^{\dagger} U \ket{i} = \bra{i} I \ket{i}=1$ and for any $i, j\in\{1, \ldots, m\}: i\neq j$,
we have $\bra{i} U^{\dagger} U \ket{j} = \bra{i} I \ket{j} = 0$.
Therefore, $U\ket{1}, \ldots, U\ket{m}$ are orthogonal vectors of length 1 and we can complete
$U$ to a $k\times k$ unitary matrix by choosing $U\ket{m+1}, \ldots, U\ket{k}$ so that they are orthogonal
(both one to another and to $U\ket{1}, \ldots, U\ket{m}$) and of length 1.
\end{proof}

\begin{lemma}
	\label{cl:alg}
	Let $A=(a_{ij})_{i\in [n], j\in[m]}$ with $\sqrt{nm}\|A\| \leq C$ and let
	\[
	p(x_1, \ldots, x_n, y_1, \ldots, y_m) = \sum_{i=1}^{n} \sum_{ j=1}^m a_{ij} x_i y_j .
	 \]
Then, there is
a quantum algorithm that makes 1 query to $x_1, \ldots, x_n$,  $y_1, \ldots, y_{m}$ and outputs 1 with probability
	\[
	r = \frac{1}{2} \lr{
		1 + \frac{p(x_1, \ldots, x_n, y_1, \ldots, y_{m})}{C}
		} .
	\]
\end{lemma}

\begin{proof}
Let $B=\sqrt{n  m}\, A$, $A=(a_{ij})$. Then,
	\[
	 \nrm B = \nrm {A}\sqrt {nm}  \leq C.
	\]
	The 1-query quantum algorithm uses a version of the well-known SWAP test \cite{BCWW}
for estimating the inner product $|\lbra \psi \ket{\psi'}|$ of two quantum states $\ket{\psi}$ and $\ket{\psi'}$.
	Our test works by preparing the state
	\begin{equation}
	\label{eq:swap}
	\frac{1}{\sqrt{2}} \ket{0}\ket{\psi} +
	\frac{1}{\sqrt{2}} \ket{1}\ket{\psi'}
	\end{equation}
	and then performing the Hadamard transformation on the first qubit and measuring the first qubit\footnote{This test is slightly different from the standard SWAP test in which one prepares both $\ket{\psi}$ and $\ket{\psi'}$ and then performs a SWAP gate conditioned by a qubit that is initially in the $\frac{1}{\sqrt{2}}\ket{0}+\frac{1}{\sqrt{2}}\ket{1}$ state. Because of this difference, we can perform the SWAP test with just 1 query instead of 2 (one for $\ket{\psi}$ and
one for $\ket{\psi'}$). Another result of this difference is that the probability of measuring 0 changes from
$\frac{1}{2} (1 + \left| \lbra \psi \ket{\psi'} \right|^2)$ for the standard SWAP test to
$\frac{1}{2} (1+ \Re \lbra \psi \ket{\psi'})$ for our test.}.
	The probability that the result of the measurement is 0
	is equal to
	\[
	r_0= \frac{1}{2} (1+ \Re \lbra \psi \ket{\psi'})
	\]
	where $\Re x$ denotes the real part of a complex number $x$.

	By Lemma \ref{th:embedB}, there is a unitary $ U $ s.t.
	 for any unit vector $\ket{y}= \sum_{i=1}^{m} \alpha_i \ket{i}$ we have
	 $U\ket{y} = \frac{B\ket{y}}{C} + \ket{\phi}$, with $ \braket i  \phi =0 $ for all $ i \in [n] $.

	The algorithm applies SWAP test to
	$\ket{x}= \frac{1}{\sqrt{n}} \sum_{i=1}^n x_i \ket{i}$
	and $U\ket{y}$,
	$\ket{y} = \frac{1}{\sqrt{m}} \sum_{i=1}^{m} y_i \ket{i}$.
	Each of those states can be prepared with one query (to $x_i$'s or $y_i$'s). Hence, we can also prepare the state (\ref{eq:swap}) with one query.
	The inner product $\lbra \psi \ket{\psi'}$ that is being estimated is equal to
	\[
	\lbra x |U \ket{y}
	=
	\frac{1}{C} \lbra x |B\ket{y}
	=
	\frac{1}{C} p(x_1, \ldots, x_n, y_1, \ldots, y_m) .
	\]
\end{proof}

Let $p(x_1, \ldots, x_n, y_1, \ldots, y_m) = \sum_{i=1}^{n} \sum_{ j=1}^m a_{ij} x_i y_j$ be the polynomial
from Definition \ref{def:bmdeg} which shows that $\bmdega_{\epsilon'}(f)=2$. Then, as we argued in subsection \ref{sec:deg2},
the matrix $A=(a_{ij})$ satisfies $\|A\|_{\infty\rightarrow 1}\leq 1$. Although this does not imply that $\|A\|$ is sufficiently small, we can preprocess the polynomial $p$ so that we achieve $\sqrt{n'm'}\|A'\|\leq K$ for the
$n'$-by-$m'$ matrix $A'$ of coefficients of the polynomial after the preprocessing.

To preprocess the polynomial, we perform an operation called \textit{variable-splitting} \cite{AA}. The operation consists of taking a variable $x_{j}$ (or $y_j$) and replacing it by $m$ variables, in the following way.  We introduce $m$ new variables $x_{l_{1}},\ldots,x_{l_{m}}$, and define $p^{\prime}$ as the polynomial obtained by substituting $\frac{x_{l_{1} }+\cdots+x_{l_{m}}}{m}$ in the polynomial $p$ instead of $x_{j}$.
If we substitute $x_{l_1}=\ldots=x_{l_m}=x_j$, $p^{\prime}$ is equal to $p(x_1, \ldots, x_n, y_1, \ldots, y_m)$.
Thus, being able to evaluate $p^{\prime}$ implies being able to evaluate $p$ (in the same sense of the word ``evaluate").

In Appendix \ref{sec:prf_ratio}, we show

\begin{lemma} \label{conj:2}
	If a polynomial
	\[
	p(x_1, \ldots, x_n, y_1, \ldots, y_m) = \sum_{i=1}^{n} \sum_{ j=1}^m a_{ij} x_i y_j .
	 \]
	 satisfies $p(x,y)\in [-1, 1]$ for all   $ x \in  \BCF^n $, $ y \in \BCF^m $, then for every $ \delta>0 $ there exists  a sequence of row and column splittings that transforms $A=(a_{ij})$ to an $n' \times m'$ matrix $A'=(a'_{ij})$ that satisfies
	\[
	\frac{
		\nrm {A'} \,  \sqrt{n' m'}
	}{
	\nrm {A'}_{\infty \to 1}
}
\leq K+\delta.
	 \]
\end{lemma}

Then, we can apply Lemma \ref{cl:alg} with $C=K+\delta$ to evaluate the polynomial
	\[
	p'(x'_1, \ldots, x'_{n'}, y'_1, \ldots, y'_{m'}) = \sum_{i=1}^{n'} \sum_{ j=1}^{m'} a_{ij} x'_i y'_j .
	 \]
for $(x'_1, \ldots, x'_{n'}, y'_1, \ldots, y'_{m'})$ 
which corresponds to the point $(x_1, \ldots, x_n, y_1, \ldots, y_m)$
at which we want to evaluate the original polynomial $p(x_1, \ldots, y_m)$.

If $p(x,y)\in[0, \epsilon']$, then Lemma \ref{cl:alg} gives $r \leq (1+\frac{\epsilon'}{K})/2$.
If $p(x,y)\in[1-\epsilon', 1]$, then  $r \geq (1+\frac{1-\epsilon'}{K})/2$.

We now consider an algorithm which outputs 0 with probability $\frac{1}{2K+1}$ and runs the algorithm
of Lemma \ref{cl:alg} otherwise (with probability $\frac{2K}{2K+1}$). Let $q$ be the probability of this algorithm outputting 1.
If $p(x,y)\in[0, \epsilon']$, then $q = \frac{2K}{2K+1} r \leq \frac{K+\epsilon'}{2K+1}$.
If $p(x,y)\in[1-\epsilon', 1]$, then $q = \frac{2K}{2K+1} r \geq \frac{K+1-\epsilon'}{2K+1}$.
Thus, we have a quantum algorithm with a probability of error which is at most $\epsilon = \frac{K+\epsilon'}{2K+1}$.

\end{proof}

\section{Results on polynomials of higher degrees}

\subsection{\texorpdfstring{$ \bmdeg$}{~bmdeg} and \texorpdfstring{$ \deg$}{~deg} vs. \texorpdfstring{$ \Q $}{Q}}
\label{sec:more}

The biggest known separation between $\deg$ and $Q$ is $Q(f)=\tilde{\Omega}(\deg^2(f))$,
recently shown by Aaronson et al. \cite{ABK} using a novel
{\em cheat-sheet} technique. We extend this result to

\begin{theorem}
\label{thm:abk}
There exists $f$ with $Q(f)=\tilde{\Omega}(\bmdeg^2(f))$.
\end{theorem}

\proof
In Appendix \ref{app:abk}.
\qed

Aaronson et al. \cite{ABK} also show a separation $Q(f)=\tilde{\Omega}(\dega(f))^4)$
which does not seem to give $Q(f)=\tilde{\Omega}(\bmdega(f))^4)$. 
(For the natural way of transforming the approximating polynomial of \cite{ABK} into a block-multilinear form,
the resulting block-multilinear polynomial $p(z^{(1)}, z^{(2)}, \ldots)$ can take values that are exponentially large
(in its degree) if the blocks $z^{(1)}, z^{(2)}, \ldots$ are not all equal.)

Because of Theorem \ref{thm:abk}, there is no transformation from a polynomial of degree $2k$ that approximates $f(x_1, \ldots, x_n)$ with error $\epsilon<1/2$ to a quantum algorithm with $k$ queries and error $\epsilon'<1/2$, with $\epsilon$ and $\epsilon'$ independent of $k$. 

However, there may be a transformation from polynomials of degree $2k$ to quantum algorithms with $k$ queries, with the error $\epsilon'=g(\epsilon, k)$
of the resulting quantum algorithm depending on $k$ but not on function $f(x_1, \ldots, x_n)$ or the number of variables $n$. 

Theorem \ref{thm:abk} implies the following limit on such transformations:

\begin{theorem}
There is a sequence of Boolean functions $f^{(1)}, f^{(2)}, \ldots$ such that,
for any sequence of quantum algorithms $\A_1, \A_2, \ldots$ computing them with $O(\bmdeg(f_i))$ queries, 
the probability of correct answer is at most
\[ \frac{1}{2} + O\left( \frac{1}{\bmdeg(f^{(i)})} \right) .\]
\end{theorem}

\begin{proof}
Let $f$ be the function from Theorem \ref{thm:abk}. 
Then, we have $\bmdeg(f)=\tilde{O}(\sqrt{n})$.

If we have a quantum algorithm $\A$ that computes a function $f$ with a probability of correct answer at least 
$\frac{1}{2}+\delta$, we can use amplitude estimation \cite{BHMT} to estimate whether $\A$ produces answer
$f=1$ with probability at least $\frac{1}{2}+\delta$ or with probability at most $\frac{1}{2}-\delta$. 
The standard analysis of amplitude estimation \cite{BHMT} shows that we can obtain 
an estimate that is correct with probability at least 
$2/3$, with $O(1/\delta)$ repetitions of $\A$. To avoid a contradiction with
$Q_{\epsilon}(f)=\Omega(n)$, we must have
\[ \frac{\sqrt{n}}{\delta} = \Omega(n) \]
which implies $\delta = O(\frac{1}{\sqrt{n}})$.
\end{proof}

A result with a weaker bound on the error is, however, possible. For example, it is possible
that $\dega_{1/2-\delta}(f)=2k$ or $\bmdega_{1/2-\delta}(f)=2k$
implies a quantum algorithm which makes $k$ queries and has the error probability 
at most $\frac{1}{2}-\Omega(\frac{\delta}{2^k})$ or at most $\frac{1}{2}-\Omega(\frac{\delta}{k^2})$. 

\subsection{Equivalence between general and block-multilinear polynomials}

By Corollary \ref{cor:equiv}, $\dega_{\epsilon}(f) \leq d$ implies $\bmdega_{\epsilon'} (f) \leq d$ 
with $\epsilon'$ that depends on $\epsilon$ and $d$ only. 
Therefore, if we want to extend the equivalence between quantum algorithms and polynomials to larger $d=O(1)$, 
it suffices to show how to transform block-multilinear polynomials into quantum algorithms. 

Also, Aaronson and Ambainis \cite{AA} showed that a quantum algorithm
which makes $d$ queries can be simulated by a classical algorithm making $O(n^{1-1/2d})$ queries, based on the following result
\begin{theorem}
	\label{thm:aa}
	\cite{AA}
	Let $ h : \mbb R^{d(n+1)} \to \mbb R $ be a block-multilinear polynomial of degree $d$ with $|h(y)|\leq 1$ for any $y\in \BCF^{d(n+1)}$.
	Then, $h(y)$ can be approximated within precision $\pm \epsilon$ with high probability, by querying $O((\frac{n}{\epsilon^2})^{1-1/d}))$ variables
	(with a big-$O$ constant that is allowed to depend on $d$).
\end{theorem}

It has been open whether a similar theorem holds for general (not block-multilinear) polynomials $h(x_1, \ldots, x_n)$. 
Aaronson and Ambainis \cite{AA} showed 
that this is true for degree 2 (using quite sophisticated tools from Fourier analysis)
but left it as an open problem for higher degrees. With Theorem \ref{thm:dz}, we can immediately resolve this problem.
\begin{corollary}
	Let $ g : \mbb R^{n} \to \mbb R $ be a polynomial of degree $d$ with $|g(y)|\leq 1$ for any $y\in \BCF^{ n}$.
	Then, $g(y)$ can be approximated within precision $\pm \epsilon$ with high probability, by querying $O((\frac{n}{\epsilon^2})^{1-1/d}))$ variables
	(with a big-$O$ constant that is allowed to depend on $d$).
\end{corollary}

\begin{proof}
	We apply Theorem \ref{thm:dz} 
	to construct a corresponding block-multilinear polynomial $h$ and then use Theorem \ref{thm:aa} to estimate
	$h$ with precision $\frac{\epsilon}{C_d}$. Since $C_d$ is a constant for any fixed $d$, we can absorb it into the big-$O$ constant.
\end{proof}

This result was independently shown by O'Donnell and Zhao \cite{DZ} (using their form of Theorem \ref{thm:dz}).

\section{Conclusions}

We have shown a new equivalence between quantum algorithms and polynomials: the existence of a 1-query quantum algorithm computing a partial Boolean function $f$ is  equivalent to the existence of a degree-2 polynomial $p$ that approximates $f$. \ Our equivalence theorem can be seen as a counterpart of the equivalence between unbounded-error quantum algorithms and threshold polynomials, proved by Montanaro et al.\ \cite{MNR}, and the equivalence between nondeterministic quantum algorithms and nondeterministic polynomials, proved by de Wolf \cite{deWolf}.

Our equivalence is, however, much more challenging to prove. A transformation from polynomials to unbounded-error or nondeterministic quantum algorithms can incur a very large loss in error probability (for example, it can transform a polynomial $p$ with error $1/3$ to a quantum algorithm $\A$ with the probability of correct answer 
$\frac{1}{2}+\frac{1}{2^n}$). In contrast, our transformation produces a quantum algorithm 
whose error probability only depends on the approximation error of the polynomial $p$ and not on the number of variables $n$.  
To achieve this, we use a relation between two matrix norms related to Grothendieck's inequality.

Our equivalence holds for two notions of approximability by a polynomial: the standard one \cite{NS} which allows arbitrary polynomials of degree 2 and the approximation by block-multilinear polynomials recently introduced by \cite{AA}. The first notion of approximability is known not to be equivalent to the existence of a quantum algorithm: there are several constructions of $f$ for which $Q_{\epsilon}(f)$ is asymptotically larger than $\deg(f)$ \cite{A03,ABK}, with 
$Q_{\epsilon}(f)=\tilde{\Omega}(\deg^2(f))$ as the biggest currently known gap \cite{ABK}. We have shown that a similar gap holds for the second notion of approximability. Thus, neither of the two notions is equivalent to the existence of a quantum algorithm in the general case.

Three open problems are:
\begin{enumerate}
\item
{\bf Equivalence between quantum algorithms and polynomials for more than 1 query?}

Is it true that quantum algorithms with 2 queries are equivalent to polynomials of degree 4?
It is even possible that quantum algorithms with $k$ queries are equivalent to polynomials of degree $2k$ for any constant $k$ - as long as the relation between the error of quantum algorithm and the error of the polynomial approximation depends on $k$, as discussed in section \ref{sec:more}.
\item
{\bf From polynomials to quantum algorithms.}

It would also be interesting to have more results about transforming polynomials into quantum algorithms, even
if such results fell short of a full equivalence between the two notions. For example, if it was possible
to transform polynomials of degree 3 into 2 query quantum algorithms this would be an interesting result, even
though it would be short of being an equivalence (since 2 query quantum algorithms are transformable into polynomials
of degree 4 and not 3).
\item
{\bf Other notions of approximability by polynomials?}

Until this work, there was a hope that the block-multilinear polynomial degree $\bmdega(f)$
may provide a quite tight characterization of the quantum query complexity $Q(f)$.
Now, we know that the gap between $\bmdeg(f)$ and $Q(f)$ can be as large as the best known
gap between $\deg(f)$ and $Q(f)$. Can one come up with a different notion of polynomial degree that
would be closer to $Q(f)$ than $\deg(f)$ or $\bmdeg(f)$?
\item
{\bf Are $deg$ and $bmdeg$ equivalent?}

For all functions $f$ that we have checked, $\dega_{\epsilon}(f)=2$ implies $\bmdega_{\epsilon}(f)=2$ with the same $\epsilon>0$. Can one prove this? More generally, does 
$\dega_{\epsilon}(f)=k$ implies $\bmdega_{\epsilon}(f)=k$ with the same $\epsilon>0$ or is there a function that can be approximated by a general polynomial of degree $k$ but not by a block-multilinear polynomial
of the same degree? As mentioned in section \ref{sec:more}, the function $f$ of \cite{ABK} with $Q(f)=\tilde{\Omega}(\dega(f))^4)$ is one candidate for separating these two notions.
\end{enumerate}

\begin{appendix}

\newpage
{\Huge Appendix}

\section{Proof of Lemma \ref{conj:2}}\label{sec:prf_ratio}
\subsection{Additional Notation}

The variables of the polynomial (\ref{eq:d2}) correspond to rows and columns of the coefficient matrix $A=(a_{ij})$, $ i \in [n] $, $ j \in [m] $.
Hence, we can reword variable-splitting in terms of rows and columns of $A$, introducing the operations of {\em row-splitting} and {\em column-splitting}.

Let $ a_{i \cdot} $ stand for the $ i $th row $ (a_{i1} , \ \ldots, \ a_{im} ) $ of $A$ and similarly $ a_{\cdot j} $ stand for the $ j $th column of $ A $.
Row-splitting (into $ k $ rows)  takes a row $a_{i \cdot}$ and replaces it with $ k $ equal rows
$a_{i \cdot}/k = (a_{i1}/k, \ldots, a_{im}/k)$.
Similarly, column-splitting  takes a column $a_{\cdot j}$ and replaces it with $ k $  equal columns  $a_{\cdot j}/ k$.

We also denote
\[
\nrm{A}_G=\sup_{k \in \mbb N}  \sup_{\substack{p_i,q_j \in \mbb R^k  \\ \forall i: \nrm{p_i}=1 \\ \forall j: \nrm{q_j}=1}}{\sum_{i,j}{a_{ij}\lra{p_i,q_j}}}.
\]
Notice that $ \nrm \cdot _G$ is a norm (and, in fact, it is the dual norm of the factorization norm $ \gamma_2 $, see, e.g., \cite{RSA:RSA20232}).

Let $  \lmax(B)$ denote the maximal eigenvalue of a square matrix $ B $; then
\begin{equation}\label{eq:max_s_val}
\nrm A ^2= \lmax \lr{A^TA}  = \lmax\lr{AA^T}.
\end{equation}

Denote
$ g(A) = \nrm A \,  \sqrt{nm}  / \nrm A_{\infty \to 1} $.
By $ \Gamma(A) $ we denote the numerator   $ \nrm A \,  \sqrt{nm}  $.

We say that a matrix $ A' $ of size $ n' \times m'$ can be obtained from $ A $ if there exists a sequence of row and column splittings
that transforms $ A $ to the matrix $ A' $; if $ A' $ can be obtained from $ A $, we denote it by $ A \obtain A' $.  Moreover, for simplicity we assume that no row or column is split  repeatedly, i.e., if a row $ a_{i \cdot} $ is split into $ k $ rows $a_{i \cdot}/k  $, then none of these obtained rows is split again.

By $ G(A) $ we denote the infimum of $ g(A') $ over all matrices $ A' $ which can be obtained from $ A $:
\[
G(A ) :=  \inf_{A' : A \obtain A'}  g(A').
\]
We have $ g(A)  \geq 1$ for all matrices $ A $. (To see this, we observe that 
$\frac{\| Ax\|_1}{\|x\|_{\infty}} \leq \frac{\sqrt{n} \|Ax\|_2}{\|x\|_2/\sqrt{m}} = \sqrt{nm} \frac{\|Ax\|_2}{\|x\|_2}$.
Taking maximums over all $x$ on both sides gives $\|A\|_{\infty\rightarrow 1} \leq \sqrt{nm} \|A\|$ which is equivalent to $g(A)\geq 1$.)
Therefore, we also have $ G(A)  \geq  1 $.

It is possible to show that the assumption that no row or column is split  repeatedly does not alter the value of this infimum; more generally, one could consider weighted splitting of rows (or columns), e.g., allowing to replace a row $ a_{i \cdot} $ with $ k $ rows $ w_j a_{i \cdot} $, $ j \in [k] $, where $ w_j $ are non-negative weights  satisfying $ w_1 +\ldots + w_k  = 1 $. Also in this case it is possible to show that the infimum of $ g(A') $ over all matrices $ A' $, yielded by  permitted splittings, has the same value as $ G(A)  $.

Let $ \mc A $ denote the class of all matrices (with real entries) which do not contain zero rows or columns. Notice that if $ A \in \mc A $ and $ A \obtain A' $, then also $ A' \in \mc A $. The class $ \mc A_{n,m}  $ contains all matrices in $ \mc A $ of size $ n \times m $.

By $\nonzVec{n} $ we denote the set of all vectors $ w \in \mbb R^n $ such that $ w_i > 0 $ for all $ i \in [n] $.

Using the introduced notation, we can restate  Lemma \ref{conj:2}:

\begingroup
\def\thelemmaPrime{\ref*{conj:2}'}
\begin{lemmaPrime}\label{th:claim0}
	For every matrix $ A $ we have
	\begin{equation}\label{eq:main}
	G(A)   =  \frac{ \nrm A_G }{\nrm A_{\infty \to 1}}  \leq K.
	\end{equation}
\end{lemmaPrime}
\endgroup
The inequality here is due to Grothendieck's inequality, see, e.g., Theorem 4 of \cite{RSA:RSA20232}. The remaining part of this section is devoted to proving the equality in \eqref{eq:main}.

\subsection{Splitting preserves the infinity-to-one norm}
Here we show that splitting rows or columns does not change the norms $ \nrm{\cdot}_{\infty \to 1} $ and $ \nrm{\cdot}_{G} $.

\begin{lemma}\label{th:claim6}
	For every matrix $ A \in \mc A$ and every $ A'  $ s.t. $ A  \obtain A' $ we have
	\[
	\nrm{A}_{\infty \to 1}
	=
	\nrm{A'}_{\infty \to 1}
	\quad
	\text{and}
	\quad
	\nrm{A}_{G}
	=
	\nrm{A'}_{G} .
	\]
\end{lemma}
\begin{proof}
	Let a matrix $ A \in \mc A_{n,m}$  be fixed. It is sufficient to show the statement for matrices $ A' $ that can be obtained by  splitting a row $ a_{i \cdot} $ of $ A $ into $ l+1 $  rows $ a_{i \cdot}/(l+1)$ (these rows are indexed by $ i $, \ldots, $ i+l $ in $ A' $).
	
	Then
	\[
	\nrm{A}_{\infty \to 1}
	=
	\max_{   x : \nrm {x}_{\infty} \leq 1  	}  \nrm{Ax}_1
	=
	\max_{   x \in \BCF^n   	}  \nrm{Ax}_1
	=
	\max_{
		\substack{
			x \in \BCF^n  \\ y \in \BCF^m }
	}x^T A y.
	\]
	Suppose that $ x \in \BCF^n , y \in \BCF^m $ are such that  $ x^T A y  = \nrm{A}_{\infty \to 1} $. Notice that
	\[
	x^T A y
	=
	\sum_{k=1}^{n} x_k a_{k \cdot} y.
	\]

	Let $ x' \in \BCF^{n+l}$ be obtained from $ x $ by replacing $ x_i $ with $ (x_i, x_i, \ldots, x_i) $ (i.e., the component $ x_i $, corresponding to the split row $ a_{i \cdot} $, is replicated $ l+1 $ times) and these components are indexed with $ i $, \ldots, $ i+l$ in $ x' $. Then
	\[
	(x')^T A' y
	=
	\sum_{k=1}^{n+l} x_k' a'_{k \cdot} y
	=
	(l+1) \cdot  x_i \frac{a_{i \cdot }}{l+1}y
	+
	\sum_{k \neq i}  x_k a_{k \cdot} y
	=
	\sum_{k=1}^{n} x_k a_{k \cdot} y
	=
	\nrm{A}_{\infty \to 1}.
	\]
	This shows that
	\[
	\nrm{A'}_{\infty \to 1} \geq \nrm{A}_{\infty \to 1}.
	\]

	Suppose that $ x \in \BCF^{n+l} , y \in \BCF^m $ are such that
	\[
	x^T A' y = \nrm{A'}_{\infty \to 1}
	\]
	and the rows $ a'_{i' \cdot} $, $ i' \in [i \tdots i+l] $, are the rows $  a_{i \cdot} / (l+1)$, obtained from $ a_{i \cdot} $.

	Let $ \tilde x \in \mbb R^n $ be such that
	\[
	\tilde x_k =
	\begin{cases}
	x_k  & k =1 ,2, \ldots, i-1 , \\
	x_{k+l}  & k =i+1 ,i+2, \ldots,n, \\
	 \frac{x_{i} +  \ldots + x_{i+l}}{l+1}, & k =i.
	\end{cases}
	\]
	Notice that
	\[
	\lrv{
		\tilde	x_i
	}
	\leq
	\frac{1}{l+1}
	\sum_{k=i}^{i+l}
	\lrv{ x_{k} }
	=1.
	\]
	Thus $ \nrm {\tilde x}_{\infty } \leq  1 $. On the other hand,
	\[
	\tilde x^T A y
	=
	\sum_{k=1}^{n} \tilde  x_k a_{k \cdot} y
	=
	\frac{\sum\limits_{k \in [i \tdots i+l]}   x_k}{l+1}    a_{i \cdot } y
	+
	\sum_{k=1}^{i-1}  x_k a_{k \cdot} y
	+
	\sum_{k=i+1}^{n}  x_{k+l} a_{k \cdot} y
	=
	\sum_{k=1}^{n+l}  x_k a'_{k \cdot} y
	= \nrm{A'}_{\infty \to 1} .
	\]
	Since
	\[
	\nrm{A }_{\infty \to 1} = \sup_{
		\substack{
			x \in \mbb R^n, y \in x \in \mbb R^m ,\\
			\nrm x_\infty \leq 1,\\
			\nrm y_\infty \leq 1
		}
	} x^T A y,
	\]
	this implies that
	\[
	\nrm{A }_{\infty \to 1}  \geq   \nrm{A'}_{\infty \to 1} .
	\]
	Hence the two norms are equal.

	Consider the norm
	\[
	\nrm{A}_G=\sup_{r \in \mbb N}  \sup_{\substack{p_k,q_j \in \mbb R^r  \\ \forall k: \nrm{p_k}=1 \\ \forall j: \nrm{q_j}=1}}{\sum_{k,j}{a_{kj}\lra{p_k,q_j}}}.
	\]
	Let  unit vectors $ p_k $, $ q_j $ (in $ \mbb R^r $ for some $ r \in \mbb N $) be fixed, $ k\in [n] $, $ j \in [m] $. Choose $ n+l $ unit vectors as follows:
	\[
	p'_k =
	\begin{cases}
	p_k, & k < i, \\
	p_{k-l} , & k = i+l+1 , \ldots, n+l, \\
	p_i, & k \in [i \tdots i+l].
	\end{cases}
	\]
	Then
	\[
	\nrm{A'}_G
	\geq
	\sum_{k,j}{a'_{kj}\lra{p'_k,q_j}}
	=
	\sum_{k,j}{a_{kj}\lra{p_k,q_j}}.
	\]
	Taking the supremum over all $ r $ and unit vectors $ p_k $, $ q_j $, we obtain
	\[
	\nrm{A'}_G
	\geq
	\nrm{A }_G.
	\]

	Let  unit vectors $ p_k $, $ q_j $ (in $ \mbb R^r $ for some $ r \in \mbb N $) be fixed, $ k\in [n+l] $, $ j \in [m] $.

	Choose $ n $ unit vectors as follows:
	\[
	\tilde p_k =
	\begin{cases}
	p_k, & k < i, \\
	p_{k+l} , & k = i+1 , \ldots, n , \\
	\frac{p_i + \ldots + p_{i+l}}{l+1} , & k =i.
	\end{cases}
	\]
	By the triangle inequality
	\[
	\nrm{ \tilde p_i} \leq \frac{\nrm {p_i} + \ldots + \nrm {p_{i+l}}}{l+1} = 1.
	 \]
	Since
	\[
	\nrm{A}_G
	=\sup_{r  \in \mbb N}  \sup_{\substack{p_k,q_j \in \mbb R^r  \\ \forall k: \nrm{p_k}=1 \\ \forall j: \nrm{q_j}=1}}{\sum_{k,j}{a_{kj}\lra{p_k,q_j}}}
	=\sup_{r \in \mbb N}  \sup_{\substack{p_k,q_j \in \mbb R^r  \\ \forall k: \nrm{p_k} \leq 1 \\ \forall j: \nrm{q_j}\leq 1}}{\sum_{k,j}{a_{kj}\lra{p_k,q_j}}},
	\]
	we have
	\[
	\sum_{k} \sum_j {a _{kj}\lra{\tilde p_k,q_j}}
	\leq
	\nrm{A}_G.
	\]

	It follows that
	\[
	\sum_{k,j}{a'_{kj}\lra{p_k,q_j}}
	=
	\sum_{k\notin [i \tdots i +l]} \sum_j{a'_{kj}\lra{p_k,q_j}}
	+
	\frac{1}{l+1}
	\sum_{k=i}^{i+l}
	\sum_j{a_{ij}\lra{p_{k},q_j}}
	=
	\sum_{k} \sum_j {a _{kj}\lra{\tilde p_k,q_j}}
	\leq
	\nrm{A}_G.
	 \]
	Taking the supremum over all $ r $ and $ p_k $, $ q_j $,  we obtain
	\[
	\nrm{A }_G
	\geq
	\nrm{A'}_G.
	\]
	Hence the two norms are equal.

\end{proof}

\subsection{Characterization of row(column)-splitting}\label{sec:4}

\begin{lemma}\label{th:claim1}
	Suppose that
	$ A \in  \mc A_{n,m} $;
	for each $ i \in [n] $ the row  $a_{i \cdot}$  is split into $ k_i $ rows
	and
	for each $j \in [m] $ the column $ a_{\cdot j} $ is split into $ l_j $ rows; the resulting matrix is denoted by $ A' $.
	
	Then $ \Gamma(A') = \nrm{\tilde A} \nrm w \nrm v $, where $\tilde{A}=(\tilde{a}_{ij})$,
	\begin{align*}
	& \tilde a_{ij} = \frac{a_{ij}}{w_i v_j}, \ i \in [n], \ j \in [m],\\
	& w_i = \sqrt {k_i}, \quad v_j  = \sqrt{l_j}.
	\end{align*}
\end{lemma}

\begin{proof}
	The matrix $ A' $  is of size $ (k_1+\ldots  + k_n) \times (l_1 + \ldots + l_m) = \nrm {w}^2 \nrm v^2$. Hence it is sufficient to show that  $ \nrm{A'} = \nrm{\tilde A} $.

	We   begin by showing this statement in case when $ l_1=l_2 = \ldots = l_m  = 1  $, i.e., only row-splitting takes place.

	Denote $ M_{i} =  a_{i \cdot}^T a_{i \cdot} $. By \eqref{eq:max_s_val},
	\[
	\nrm {\tilde A}^2
	=
	\lmax(\tilde A^T\tilde A),
	\qquad
	\nrm { A'}^2
	=
	\lmax({A'}^T  A'),
	\]
	Notice that
	\[
	\tilde A^T \tilde A
	=
	\begin{pmatrix}
	w_1^{-1} a_{1 \cdot}^T
	&
	w_2^{-1}  a_{2 \cdot}^T
	&
	\ldots
	&
	w_n^{-1} a_{n \cdot}^T
	\end{pmatrix}
	\begin{pmatrix}
	w_1^{-1} a_{1 \cdot}
	\\
	w_2^{-1}  a_{2 \cdot}
	\\
	\ldots
	\\
	w_n^{-1} a_{n \cdot}
	\end{pmatrix}
	=
	\sum_{i=1}^{n}  w_i^{-2} M_{i}.
	\]
	Similarly it can be obtained that
	\[
	{A'}^T  A'
	=
	\sum_{i=1}^{n}   \sum_{j=1}^{k_i}  \frac{1}{k_i^2}\, M_{i}  .
	\]
	Since
	\[
	\sum_{i=1}^{n}   \sum_{j=1}^{k_i}  \frac{1}{k_i^2}\, M_{i}
	=
	\sum_{i=1}^{n}    \frac{1}{k_i}\, M_{i}
	=
	\sum_{i=1}^{n}  w_i^{-2} M_{i},
	\]
	we conclude that
	\[
	{A'}^T  A' =\tilde A^T \tilde A,
	\]
	which implies $  \nrm  {\tilde A}= \nrm  { A'} $.
	
	Now consider the case of arbitrary $ l_j \in \mbb N $. Denote by $ B $ the $ n \times (l_1  + \ldots+ l_m) $ matrix, obtained from $ A   $ by splitting each of its columns $ a_{\cdot j} $   into $ l_j $ columns. Then $ A \obtain B \obtain A' $. By the previous arguments,
	\[
	\nrm {A'} = \nrm{\tilde B},
	 \]
	where $  \tilde B$ is $ \tilde B$ is $ n \times (l_1  + \ldots+ l_m)  $ matrix  with $ i $th row equal to
	\[
	\begin{pmatrix}
	\underbrace{ \frac{a_{i1}}{l_1 \sqrt {k_i}}   }_{\text{repeated }l_1 \text{ times}}
	&
	\underbrace{ \frac{a_{i2}}{l_2 \sqrt {k_i}}   }_{\text{repeated }l_2 \text{ times}}
	&
	\ldots
	&
	\underbrace{ \frac{a_{im}}{l_m \sqrt {k_i}}   }_{\text{repeated }l_m \text{ times}}
	\end{pmatrix}.
	 \]
	Then the transpose of $ \tilde B $ can be obtained from the $ m \times n $  matrix $ C = \lr{C_{ji}} $,
	\[
	C_{ji} = \frac{a_{ji}}{\sqrt {k_i}}, \quad i \in [n], \ j\in [m],
	 \]
	by splitting the $ j$th row of $  C$ into $ l_j $ rows.

	By previous argument,
	\[
	\nrm{\tilde B^T} = \nrm{\tilde C},
	 \]
	where $ \tilde C = \tilde A^T $.
	Thus we conclude
	\[
	\nrm {A'} = \nrm{\tilde B}=	\nrm{\tilde B^T} = \nrm{\tilde A^T} = \nrm {\tilde A}.
	 \]
\end{proof}

This shows that $ \Gamma(A') $,  for every matrix $ A' $ which can be obtained from $ A $ by splitting rows/columns, can be characterized by vectors $ w $, $ v $ (s.t. the squares of components of $ w,v $ are rational numbers). The   converse  is also true:

\begin{lemma}\label{th:claim3}
	Suppose that
	$ A  \in  \mc A_{n,m}$  but vectors $ w \in \nonzVec{n}  $, $ v \in \nonzVec{m}$ are such that $ w_i^2 \in \mbb Q $, $ v_j^2 \in \mbb Q $ for all $ i,j $. Then there exist numbers $ k_i \in \mbb N $ and $ l_j \in \mbb N $ such that splitting $ A $'s   $ i $th  row  $a_{i \cdot}$   into $ k_i $ rows
	and
	the $ j $th column $ a_{\cdot j} $  into $ l_j $ rows yields a matrix $ A' $ such that $ \Gamma(A') = \nrm{\tilde A} \nrm w \nrm v $ where $\nrm{\tilde A} = ({\tilde a_{ij}})$, $  \tilde a_{ij} := \frac{a_{ij}}{w_i v_j} $.
\end{lemma}
\begin{proof}
	First note that the statement is true if $ w_i^2  \in \mbb N $ and $ v_j^2 \in \mbb N $ for all $ i,j $, since then one takes $ k_i = w_i^2 $ and $ l_j = v_j^2 $.
	
Since $ w_i^2 \in \mbb Q $, $ v_j^2 \in \mbb Q $, we have	
$ w_i^2 = \frac{p_i}{p'_i} $ and $ v_j^2 =  \frac{q_j}{q'_j} $ for some
natural numbers $ p_i$, $ p'_i $, $ q_j $ and $ q_j' $.
	Denote $ P = \prod_i p'_i $ and $ Q = \prod_j q'_j $. Let $ \hat w_i =  w_i \, \sqrt {P} $,  $ \hat v_j =  v_j\, \sqrt {Q} $ and  $ \hat  A = (\hat a_{ij}) $, where
	\[
	\hat a_{ij}  =  \frac{a_{ij}}{\hat w_i \hat v_j} = \frac{\tilde a_{ij}}{\sqrt {PQ}}.
	\]
	Then
	\[
	\nrm{ \hat A } = \frac{1}{\sqrt{PQ}}  \nrm {\tilde A},
	\quad
	\nrm {\hat w} = \sqrt{P} \nrm w,
	\quad
	\nrm {\hat v} = \sqrt{Q} \nrm v.
	\]
	Thus
	\[
	\nrm{\tilde A} \nrm w \nrm v  = \nrm{\hat A} \nrm {\hat w}\nrm {\hat v} .
	\]
	Moreover, $ \hat w_i^2 \in \mbb N $, $ \hat v_j^2 \in \mbb N $, thus one can take $ k_i = \hat w_i^2 $ and $ l_j = \hat v_j^2 $. Now, by performing the corresponding row/column splitting, one obtains a matrix $ A' $ satisfying
	\[
	\Gamma(A')  = \nrm{\hat A} \nrm {\hat w}\nrm {\hat v}= \nrm{\tilde A} \nrm w \nrm v.
	\]
\end{proof}

We can consider an even more general situation:
\begin{lemma}\label{th:claim4}
	Suppose that
	$ A  \in  \mc A_{n,m}$ and $ w \in \nonzVec{n}   $, $ v \in \nonzVec{m}$.
	
	Then there exist sequences $ (k_{i,N} )_{N}  \subset \mbb N $ and $(l_{j,N} )_{N}  \subset  \mbb N $ such that
	\[
	\lim_{N \to \infty}   \Gamma(A_N') = \nrm{\tilde A} \nrm w \nrm v.
	\]
	
	Here by $ \tilde A $ we denote  the matrix with components $  \tilde a_{ij} = \frac{a_{ij}}{w_i v_j} $, but $ A_N' $
	stands for the matrix which is obtained from $ A $ by splitting its   $ i $th  row  $a_{i \cdot}$   into $ k_{i,N} $ rows
	and
	the $ j $th column $ a_{\cdot j} $  into $ l_{j ,N}$ rows.
\end{lemma}
\begin{proof}
	We choose two sequences of vectors $w^{(1)}, w^{(2)}, \ldots$ and
	$v^{(1)}, v^{(2)}, \ldots$ so that
	$w^{(N)} \in Q^{n}_{+}$ and $w=\lim_{N\rightarrow\infty} w^{(N)}$
	and similarly for $v^{(N)}$ and $v$.
	Let $\tilde{A}^{(N)}$ be a matrix with entries
	$\tilde{a}^{(N)}_{ij}=\frac{a_{ij}}{w_i v_j}$.
	
	Then, by Lemma \ref{th:claim3}, there are matrices $A'_N$ such that
	$\Gamma(A'_N) = \| \tilde{A}^{(N)} \| \|w^{(N)}\| \|v^{(N)}\|$.
	Let $k_{i, N}$ and $l_{i, N}$ be the values of $k_i$ and $l_i$ in the application
	of Lemma \ref{th:claim3}.
	By continuity, if $N\rightarrow\infty$, we have
	$\|w^{(N)}\| \rightarrow \|w\|$, $\|v^{(N)}\| \rightarrow \|v\|$,
	$\| \tilde{A}^{(N)} \|\rightarrow \|\tilde{A}\|$.
	
	Hence, $\lim_{N \to \infty}   \Gamma(A_N') = \nrm{\tilde A} \nrm w \nrm v$.
\end{proof}

Suppose that $ A   \in  \mc A_{n,m}$  and $ w \in \nonzVec{n}   $, $ v \in \nonzVec{m}$  are fixed. Let $ \tilde A $ be  the matrix with components
\[
\tilde a_{ij} = \frac{a_{ij}}{w_i v_j}.
\]
Notice that $ \tilde A = D_{w}^{-1} A D_v^{-1}   $. Denote
\[
F_A(w,v)
=
\nrm{D_{w}^{-1} A D_v^{-1}} \nrm{w} \nrm v.
\]
Then Claims  \ref {th:claim1} and \ref {th:claim4}   together imply that
\[
\inf _{A' :  A \obtain A'}
\Gamma(A')
=
\inf_{
	\substack{
		w \in \nonzVec{n}\\
		v \in \nonzVec{m}
	}
}
F_A(w,v).
\]
Denote the latter infimum with $ F_A^T  $.
In view of Lemma  \ref {th:claim6} this means that
\begin{equation}\label{eq:main2}
G(A)
=
\frac{\inf _{A' :  A \obtain A'}
	\Gamma(A')}{\nrm{A}_{\infty \to 1}}
=
\frac{
	F_A^T
}{\nrm{A}_{\infty \to 1}}.
\end{equation}

\subsection{Proof of Lemma \ref*{th:claim0}}\label{sec:6}

We recall the following characterization of matrices with $ \nrm A_G \leq 1 $; for a proof, see \cite[p. 239]{Pisier}.

\begin{lemma}\label{th:claim7}
	For every matrix $ A $ (of size $ n\times n $), the inequality $ \nrm A_G \leq 1 $ holds iff there is a matrix $ \tilde A $ (of size $ n\times n $) and vectors  $ w, v \in \mbb R^n $ with non-negative components s.t.
	$ \nrm w = \nrm {v} = 1 $, $ \nrm {\tilde A} \leq 1 $ and for all $ i,j \in [n] $: $ a_{ij}  = \tilde a_{ij} w_i v_j  $.
\end{lemma}

From this it is easy to obtain the following:
\begin{lemma}\label{th:claim8}
	For every matrix $ A  \in \mc A_{n,n} $   there exists a matrix $ \tilde A  \in \mc A_{n,n}$  and vectors  $ w, v \in \nonzVec{n}  $ s.t.
	$ \nrm w = \nrm {v} = 1 $, $ \nrm {\tilde A} = \nrm{A}_G $ and  $ \tilde A = D_w^{-1} A D_v $. Moreover, $ w $ and $ v $ minimize the function $ F_A(\cdot, \cdot) $, i.e.,
	\[
	F^T_A  =  \nrm {\tilde A}\nrm w   \nrm {v}=  \nrm{A}_G.
	\]
\end{lemma}
\begin{proof}
	Suppose that a matrix $ A \in \mc A_{n,n}$  is scaled so that $ \nrm {A}_G  = 1 $.
	
	From Lemma \ref{th:claim7} the existence of $ \tilde A $ with $ \nrm {\tilde A} \leq 1 $ and $ w ,  v  \in \nonzVec{n} $ with $ \nrm w = \nrm v = 1 $ follows. 
Notice that $ w_i \neq 0 $ and $ w_j' \neq 0 $ for all $ i,j $, since otherwise  $ A \notin \mc A $. Similarly, also $  \tilde A \in \mc A_{n,n} $ must hold.

	We claim that    $  \nrm{\tilde A} = 1 $. Assume the contrary,  $  \nrm{\tilde A} = c \in (0,1)$.
	
	Let $ \tilde B $ be a $ n \times n $ matrix with $ \tilde b_{ij} = \tilde a_{ij}   / c  $,
	then $ \nrm {\tilde B} = 1 $ and by Lemma \ref{th:claim7} we have $ \nrm B_G \leq 1 $, where $ B = A/c $. But then $ \nrm {A}_G \leq c < 1 $, a contradiction.	Thus $ \nrm {\tilde A}_G  = 1 $.
	
	To prove the second part of the statement,  suppose that there are unit vectors $ \hat{w},\hat{v} \in  \nonzVec{n}  $ such that $ F_A(\hat{w},\hat{v}) =  s <1$.  Let $ \tilde X = D_{\hat w}^{-1} A D^{-1}_{\hat v}/s  $, then $ \nrm{\tilde X} =1 $. By Lemma \ref{th:claim7} we have $ \nrm X_G \leq 1 $, where $ X  = A/s$. But then $ \nrm A_G \leq s < 1 $, a contradiction.
\end{proof}

\begin{proof}[Proof of Lemma \ref{th:claim0}]
	\textbf{The case of $ A \in \mc A $}.
	
	Notice that
	\[
	\inf _{A' :  A \obtain A'}
	\Gamma(A')
	=
	\inf _{A' :  A'' \obtain A'}
	\Gamma(A'),
	\]
	where $ A'' $ is any matrix s.t. $ A \obtain A'' $. This means that  $ F_{A} ^T = F_{A'} ^T $, if $ A \obtain A'  $.
	To apply Lemma \ref{th:claim8}, transform $ A $ into  a square matrix  $ A' $ by splitting a row or  a column. Then
	\[
	F_A^T
	=
	F_{A'}^T
	\overset {\text{Lemma \ref{th:claim8}}}
	=
	\nrm{A'}_G
	\overset {\text{Lemma \ref{th:claim6}}}
	=
	\nrm{A}_G
	\]
	and, by \eqref{eq:main2},
	$ G(A) =  \nrm A_G/ \nrm A_{\infty \to 1}  $,
	proving \eqref{eq:main} for all $ A \in \mc A $.

	It remains to show that
	\eqref{eq:main} holds
	for all matrices $ A $.

	\textbf{The case of $ A \notin \mc A $}.
	Suppose that  $ A $ is a $ n\times m $ matrix and there are $ k$  zero rows and $ l $  zero columns. W.l.o.g. assume the non-zero rows/columns are the first, then
	\[
	A
	=
	\begin{pmatrix}
	\hat A & 0_{n-k,l} \\
	0_{k,m-l}   & 0_{k,l}
	\end{pmatrix},
	\]
	where $ \hat A \in \mc A_{n-k,m-l} $ (and $ 0_{a,b} $ stands for the zero matrix of size $ a \times b $).
	Notice that
	\[
	g(\hat A) = \frac{\nrm {\hat A}  \sqrt{ (n-k)(m-l) }}{\nrm {\hat A}_{\infty \to 1}}
	=
	\frac{\nrm { A}  \sqrt{ (n-k)(m-l) }}{\nrm {  A}_{\infty \to 1}}
	<
	\frac{\nrm { A}  \sqrt{nm}}{\nrm {  A}_{\infty \to 1}}
	=
	g(A).
	\]
	
	By the previous case, we have
	$
	G(\hat A)
	=
	\nrm {\hat A}_G/\nrm {\hat A}_{\infty \to 1}
	=
	\nrm A_G/\nrm A_{\infty \to 1}$.

	Clearly, for  every $ A' $ with $ A \obtain A' $  we have $ \hat A' $ s.t. $ \hat A \obtain \hat A' $  and $ g(\hat A' )  \leq g(A') $ (take $ \hat A' $ to  be the   minor of $ A' $, obtained by skipping all zero rows or columns).  Then
	$
	G(\hat A) \leq g(\hat A') < g(A')$.
	Taking infimum over all $ A' $ s.t. $ A \obtain A' $, inequality $ G(\hat A) \leq G(A ) $ follows.
	
	On the other hand,
	for every $ \hat A' $ s.t. $ \hat A \obtain \hat A' $ we have a sequence $ (A_N)_{N \in \mbb N}$ with $ A \obtain A_N $ for all $ N $ and $  \lim_{N \to \infty} g(A_N)  = g(\hat A' ) $: take the matrix
	\[
	B
	=
	\begin{pmatrix}
	\hat A' & 0_{p,l} \\
	0_{k,q}   & 0_{k,l}
	\end{pmatrix},
	\]
	where $ \hat A' $ is of size $ p \times q $ (i.e., 
$ B  $ is the matrix obtained by splitting the non-zero part of $ A $ in the same way how we split $\hat{A}$ to obtain $\hat{A'}$). Then the matrix $ A_N $ is obtained by splitting each row $ b_{i \cdot} $, $ i \in [p]  $ of $B$, and each column $ b_{\cdot j} $, $ j \in [q]$ of $B$ into $ N $ rows/columns. We have $A\obtain B \obtain A_N$ and
the resulting matrix $ A_N $ is of  size $ (Np + k) \times (Nq + l) $.
We denote the upper $ Np \times Nq $ submatrix of $A_N$ by $ B_N $. Then  $ B_N =  \frac{1}{N^2}  \hat A' \otimes J_{N,N} $, where $ J_{N,N} $ is the $ N \times  N$ all-1 matrix.
	
	We have
	\begin{align*}
	& \nrm{A_N} = \nrm {B_N} = \frac{\nrm {\hat A'} }{N};\\
	& \nrm{A_N}_{\infty \to 1} = \nrm {B_N}_{\infty \to 1} = \nrm {\hat A'}_{\infty \to 1};
	\\
	& g(A_N) =
	\frac{\nrm{A_N} \sqrt{(Np +k)\cdot (Nq+l)} }{ \nrm{A_N}_{\infty \to 1}}
	=
	\frac{
		\nrm {B_N} \sqrt{(Np +k)\cdot (Nq+l)}
	}{
	\nrm {B_N}_{\infty \to 1}
}
\\ & =
\frac{\nrm {\hat A'} \sqrt{pq}}{ \nrm {\hat A'}_{\infty \to 1}}
\cdot
\sqrt{
	\frac{Np +k }{N p} \cdot \frac{Nq +l}{Nq}
}
=
g(\hat A')
\sqrt{
	\lr{1 + \frac{c_1}{N}} \lr{1 + \frac{c_2}{N}}
},
\end{align*}
where $ c_1 = k/p $, $ c_2  = l/q $.

We see that
$
G(A) \leq
\lim_{N \to \infty} g(A_N)  = g(\hat A')$.
Taking infimum over all $ \hat A' $ s.t. $ \hat A \obtain \hat A' $, inequality $ G(\hat A) \geq G(A ) $ follows. Hence the two quantities  must be equal.
\end{proof}

\section{Proof of Theorem \ref{thm:abk}}
\label{app:abk}

We use the notion of {\em certificate complexity.}
Let $C$ be an assignment of values $C:S\rightarrow \{0, 1\}$ for some $S\subseteq [n]$.
We say that $x=(x_1, \ldots, x_n)$ is {\em consistent} with $C$ if it satisfies $x_i=C(i)$ for all $i\in S$.
We say that $C$ is a {\em certificate} for $f$ on an input $x$ if $x$ is consistent with $C$ and, for any $y\in \{0, 1\}^n$ that is consistent with $C$,
we have $f(y)=f(x)$. 

The certificate complexity of $f$ on an input $x$ (denoted by $C(f, x)$) is the smallest $|S|$ in a certificate $C$ for $f$ on the input $x$.
The certificate complexity of $f$ (denoted $C(x)$) is the maximum of $C(f, x)$ over all $x\in\{0, 1\}^n$.
(For more information on the certificate complexity and its connections to other complexity measures, we refer the reader to
the survey by Buhrman and de Wolf \cite{BW}.)

We use the same function as in the $Q(f)=\tilde{\Omega}(\deg^2(f))$ result of Aaronson et al. \cite{ABK}.
The construction of this function \cite{ABK} starts by designing
a function $g:\{-1, 1\}^n \rightarrow \{0, 1\}$ with $Q(g)=\tilde{\Omega}(n)$ and $C(g)=\tilde{O}(\sqrt{n})$.
(We omit the definition of $g$ because $Q(g)=\tilde{\Omega}(n)$ and $C(g)=\tilde{O}(\sqrt{n})$ are the only
properties of $g$ that we use.)

Then, they define $f$ as follows:
\begin{enumerate}
\item
The first $c=10 n \log n$ input variables of $f$ are interpreted as $c$ inputs $x^{(1)}\in\{0, 1\}^n, \ldots, x^{(c)}\in\{0, 1\}^n$
to the function $g$.
\item
These input variables are followed by $2^c$ groups of variables $y^{(m)}$, $m\in\{0, 1\}^c$, with each
group containing $c C(g) \log n$ variables.
The content of each $y^{(m)}$ is interpreted as descriptions for $c$ sets $S_1, \ldots, S_c\subseteq [n]$ with $|S_j|=C(g)$. 
A set $S_j$ is interpreted as a sequence of indices for $C(g)$ variables for the function $g(x^{(j)})$.
\item
$f=1$ if and only if, for some $m\in\{0, 1\}^c$, the group $y^{(m)}$ 
contains descriptions for sets $S_i$ such that, for each $i\in[c]$, the variables $x^{(i)}_j$, $j\in S_i$ form
an $m_i$-certificate.
\end{enumerate}

As shown in \cite{ABK}, $f$ satisfies $Q(f)=\tilde{\Omega}(n)$ and $\deg(f)=\tilde{O}(\sqrt{n})$.
A polynomial $p$ of degree $\tilde{O}(\sqrt{n})$ that represents $f$ can be constructed as follows:
\begin{enumerate}
\item
$p=\sum_{m\in\{0, 1\}^c} p_m$; 
\item
$p_m = \sum_{S_1, \ldots, S_c} p_{m, S_1, \ldots, S_c}$, with the summation over all tuples
$(S_1, \ldots, S_c)$ such that, for all $i\in[c]$, $S_i$ is a possible certificate for $g(x)=m_i$;
\item
$p_{m, S_1, \ldots, S_c} = q_{m, S_1, \ldots, S_c} \prod_{i=1}^c r_{i, m_i, S_i}$;
\item
$q_{m, S_1, \ldots, S_c} =1$ if the contents of $y^{(m)}$  describe sets $S_1, \ldots, S_c$ and $q_{m, S_1, \ldots, S_c} =0$ otherwise;
\item
$r_{i, m_i, S_i} = 1$ if the values of variables $x^{(i)}_j, j\in S_i$ certify that $g(x^{(i)})=m_i$ and $r_{i, m_i, S_i} = 0$ otherwise.
\end{enumerate} 
In the non-block-multilinear case, $q_{m, S_1, \ldots, S_c}$ is the product of $\frac{1+y^{(m)}_i}{2}$'s (for $i$'s
where we need $y^{(m)}_i=1$) and $\frac{1-y^{(m)}_i}{2}$'s (for $i$'s where we need $y^{(m)}_i=-1$).
$r_{i, m_i, S_i}$ is constructed similarly, by taking a product of
$\frac{1+x^{(i)}_j}{2}$'s and $\frac{1-x^{(i)}_j}{2}$'s for $j\in S_i$, to obtain the condition that $x^{i}_j$
take the values that are necessary so that $x^{i}_j$, $j\in S_i$ certify $g(x^{(i)})=m_i$.

We now modify this construction to obtain $\bmdeg(f)=\tilde{O}(\sqrt{n})$. 
Our polynomial has blocks of variables $z^{(i)}$, for $i\in[c C(g) (\log n+1)]$, with each $z^{(i)}$
consisting of a variable $z^{(i)}_0$, $c$ subblocks $x^{(i, 1)}, \ldots, x^{(i, c)}$ and $2^c$ subblocks 
$y^{(i, m)}$ for $m\in\{0, 1\}^c$.

The structure of the polynomial $p$ stays the same and we only modify the constructions of 
$q_{m, S_1, \ldots, S_c}$ and $r_{i, m_i, S_i}$.
To construct $q_{m, S_1, \ldots, S_c}$, we use the first $c C(g) \log n$ blocks $z^{(i)}$, taking 
the value of $y^{(m)}_i$ from the $i^{\rm th}$ block and using $z^{(i)}_0$ instead of 1 in the terms
$\frac{1\pm y^{(m)}_i}{2}$.

To construct $r_{i, m_i, S_i}$, we use $z^{(k)}$ for $k\in\{(c \log n+(i-1)) C(g)+1,
\ldots, (c\log n+i) C(g)\}$ and take $r_{i, m_i, S_i}$ to be the average of the 
desired product of $\frac{z^{(k)}_0+x^{(k, i)}_j}{2}$'s and $\frac{z^{(k)}_0-x^{(k, i)}_j}{2}$'s
over all the ways how one could use one term per block $z^{(k)}$.  

It is easy to see that, if all blocks $z^{(i)}$ contain the same assignment $z$, 
then $p(z, \ldots, z)$ is the same polynomial as in the non-block-multilinear case and is equal to $f(z)$.
We now show that $|p|\leq 1$ for any choice of $z^{(1)}, z^{(2)}, \ldots$ in which all the variables are in $\{-1, 1\}$.

For each $m$, all polynomials $q_{m, S_1, \ldots, S_c}$ use the same variables $z^{(i)}_0$ and $y^{(i, m)}_i$ and are defined
so that, for any choice of values for $z^{(i)}_0$'s and $y^{(i, m)}_i$'s, at most one of $q_{m, S_1, \ldots, S_c}$ 
is $\pm 1$ and the rest are 0. 
Let $S_{m, 1}, \ldots, S_{m, c}$ be the sets for which $q_{m, S_{m, 1}, \ldots, S_{m, c}}=\pm 1$ 
(if such sets exist).
Then, $p(z^{(1)}, \ldots, z^{(c C(g) (\log n+1))})$ is equal to the sum
\begin{equation}
\label{eq:imp} 
\sum_{m\in\{0, 1\}^c} a_m \prod_{i=1}^c  r_{i, m_i, S_{m, i}} 
\end{equation}
for some choice of signs $a_m \in\{-1, 1\}$.
We show

\begin{lemma}
\label{lem:abk}
Let 
$S_{m, i}$, $m\in\{0, 1\}^c$, $i\in [c]$ be such that $S_{m, i}$ is an $m_i$-certificate for the function $g$.
Then,
\[ \left| \sum_{m\in\{0, 1\}^c} a_m \prod_{i=1}^c  r_{i, m_i, S_{m, i}}\right|  \leq 1 \]
for any choice of signs $a_m \in\{-1, 1\}$.
\end{lemma}

\proof
By induction on $c$. For $c=1$, this simplifies to
\begin{equation}
\label{eq:pm} -1 \leq a_0 r_{1, 0, S_{0, 1}} + a_1 r_{1, 1, S_{1, 1}} \leq 1 
\end{equation}
when $S_{0, 1}$ is a set of variables for a 0-certificate and 
$S_{1, 1}$ is a set of variables for a 1-certificate.
Since a 0-certificate and a 1-certificate cannot be true at the same time, there must
be $j\in S_{0, 1}\cap S_{1, 1}$ with $x_j$ taking one value in the 0-certificate and another
value in the 1-certificate. 

Let $p_0$ be the probability that, when we choose a block $z^{(i)}$ randomly among the blocks that are
used to define $r_{1, m_1, S_1}$'s, we get the value of $x^{(i, 1)}_j$ which matches the 0-certificate.
Then, the probability of getting the value that matches the 1-certificate is $1-p_0$ and 
we get that $\lrv{r_{1, 0, S_{0, 1}}}\leq p_0$ and $\lrv{r_{1, 1, S_{1, 1}}} \leq 1-p_0$.
This implies (\ref{eq:pm}) for any choice of signs $a_0, a_1\in\{-1, 1\}$.

For $c>1$, we can use the same argument to show that, for any $m\in\{0, 1\}^{c-1}$,
we have $r_{c, 0, S_{m0}}\leq p_m$ and $r_{c, 1, S_{m1}}\leq 1-p_m$ for some $p_m$ that depends on $m$.
Therefore, the sum of Lemma \ref{lem:abk} is upper bounded by 
\[ \sum_{m\in\{0, 1\}^{c-1}} \left( p_m a_{m0} \prod_{i=1}^{c-1} r_{i, m_i, S_{m0, i}} +
(1-p_m) a_{m1} \prod_{i=1}^{c-1} r_{i, m_i, S_{m1, i}} \right) .\]
We can express this sum as a probabilistic combination of sums 
\begin{equation}
\label{eq:reduced} \sum_{m\in\{0, 1\}^{c-1}} a_m  \prod_{i=1}^{c-1} r_{i, m_i, S_{m, i}} 
\end{equation}
where each $S_{m, i}$ is either $S_{m0, i}$ or $S_{m1, i}$ and each $a_m$ is either $a_{m0}$ or $a_{m1}$.
Each of sums (\ref{eq:reduced}) is at most 1 in absolute value by the inductive assumption.
\section{Proof of Theorem \ref{thm:dz}}
\label{sec:higher}

We expand the polynomial $ p$ in the Fourier basis as
\[
p(x) = \sum_{\substack{
		T \subset [n]: \\ 
		\lrv T \leq d
	}}
	\hat p_T \chi_T(x),
	\]
	where $ \chi_T(x) = \prod_{i \in T} x_i  $.
	For each $\chi_T(x)$, we define a corresponding block multilinear polynomial
	\[ 
	\chi'_T\lr{z^{(1)}, z^{(2)}, \ldots, z^{(d)}  }
	=
	\frac{(d-|T|)!}{d!} 
	\sum_{
		\substack{
			B \subset [d] :\\
			\lrv B = \lrv T
		}
	}
	\sum_{
		\substack{
			b :\\
			b : B \to T\\
			b \text{ -- bijection}
		}
	}
	\prod_{j \in B}  z^{(j)}_{b(j)} 
	\prod_{k \in [d]\setminus B}  z^{(k)}_{0}.
	\]
	We then take 
	\[ 
	\tilde p \lr{z^{(1)}, z^{(2)}, \ldots, z^{(d)}  } 
	= \sum_{T} \hat p_T \chi'_T \lr{z^{(1)}, z^{(2)}, \ldots, z^{(d)}  }.
	\]
	
	If we set
	\[ 
	\hat z^{(j)}_k = 
	\begin{cases}
	1, & k = 0;\\
	x_k, & k \in [n]
	\end{cases}
	\]
	for some $x\in \mbb R^n$, we get
	\begin{multline*}
	\chi'_T \lr{\hat z^{(1)}, \hat z^{(2)}, \ldots, \hat z^{(d)}  }
	=
	\frac{(d-|T|)!}{d!} 
	\sum_{
		\substack{
			B \subset [d] :\\
			\lrv B = \lrv T
		}
	}
	\sum_{
		\substack{
			b :\\
			b : B \to T\\
			b \text{ -- bijection}
		}
	}
	\prod_{j \in B}  x_{b(j)}
	=
	\frac{(d-|T|)!}{d!} \,
	\binom d {|T|} 
	\, |T|! 
	\prod_{j \in T}  x_{s}
	= \chi_T(x) 
	\end{multline*}
	and, therefore, $\tilde p \lr{1, x, 1, x, \ldots, 1, x } = p(x)$.

	Since $ \tilde p $ is a $ d $-linear map from $ E^d $ to $ \mathbb{R} $, where $ E = \mathbb{R}^{n+1} $, it satisfies \cite[Eq.7]{Thomas} the polarization identity 
	\[ 
	d!  \tilde p \lr{z^{(1)} , z^{(2)}   , \ldots ,  z^{(d)} }
	=
	\sum_{\substack{T \subset [d]: \\ T \neq \emptyset}} (-1)^{d-\lrv T} \, \hat p \lr{\sum_{j \in T}  z^{(j)}  },  \qquad z^{(1)}, \ldots ,  z^{(d)}  \in \mbb R^{n+1},
	\]
	where $ \hat p(z):= \tilde p(z,z,\ldots,z)$, $ z \in \mathbb{R}^{n+1} $. Since $ \tilde p $ is a homogeneous polynomial, also $ \hat p $ is homogeneous, thus we obtain
	\begin{equation}\label{eq:tmp}
	\tilde p \lr{z^{(1)} , z^{(2)}   , \ldots ,  z^{(d)} }
	=
	\frac{1}{d!}
	\sum_{\substack{T \subset [d] :\\ T \neq \emptyset}} 
	(-1)^{d-\lrv T}  \lrv T^d \, \hat p \lr{ \frac{\sum_{j \in T}  z^{(j)} }{\lrv T} },  
	\qquad z^{(1)}, \ldots ,  z^{(d)}  \in \mbb R^{n+1}.
	\end{equation}

	To show the bound on $ \tilde p $, fix $z^{(1)}, \ldots ,  z^{(d)}  \in \BCF^{n+1}  $;	 we can assume that $z^{(j)}_0=1$ for all $j\in [d]$. (If $z^{(j)}_0=-1$, we multiply all $z^{(j)}_i$ by $ -1 $ and $|\tilde p\lr{z^{(1)}, z^{(2)}, \ldots, z^{(d)}  }|$ stays unchanged.)
	For every $ T \subset [d] $, $ T \neq \emptyset $, denote $ (x^T_0, \ldots, x^T_n):= \frac{1}{\lrv T}\sum_{j \in T}  z^{(j)}   $; then
	\begin{itemize}
		\item $ x^T_0 = 1 $;
		\item $   x^T := (x^T_1, \ldots, x^T_n) \in [-1;1]^n  $.
	\end{itemize}
	
	Since $ x^T_0 = 1 $, from the first part of the proof we have $ \hat p(x^T_0, \ldots, x^T_n)  = p(x^T)$.
	Therefore, from \eqref{eq:tmp} we obtain
	\begin{equation}\label{eq:tmp1}
	\tilde p \lr{z^{(1)} , z^{(2)}   , \ldots ,  z^{(d)} }
	=
	\frac{1}{d!}
	\sum_{\substack{T \subset [d] :\\ T \neq \emptyset}} (-1)^{d-\lrv T}  \lrv T^d \, p \lr{ x^T }.
	\end{equation}
	
	Since $ p $ is multilinear, its maximum over $  [-1;1]^{n} $ coincides with its maximum over    $   \BCF^{n} $; i.e., $ \lrv {p(x)} \leq 1 $ for all $ x \in [-1;1]^n $, thus also $ \lrv{p \lr{ x^T }}  \leq 1 $ for all $ T  $.
	We conclude that
	the value of $\lrv {\tilde p \lr{z^{(1)} , z^{(2)}   , \ldots ,  z^{(d)} }} $ is at most
	\[ 
	C_d
	\leq 
	\frac{1}{d!}
	\sum_{
		\substack{
			T \subset [d]: \\
			T  \neq \emptyset
		}
	}
	\lrv{T}^{d}
	=
	\frac{1}{d!} \sum_{s=1}^{d} \binom{d}{s}s^d := B(d).
	\]

	Let us show that $ B(d)  = \Theta \lr{\frac{\alpha^d}{\sqrt d}} $, where 
		$\alpha = 1/W(\exp(-1))  \approx  3.5911$
		and $ W $ stands for the (primary branch of) Lambert W function.
	Let $ \beta = 1/\alpha = W(1/e) $. It is known \cite{Kot} that
	\[ 
	\sum_{s=1}^{d} \binom{d}{s}s^d
	\sim
	\frac{1}{\sqrt{  1 + \beta  }}
	\lr{\frac{d}{e \beta}  }^d   .
	\] 
	By Stirling's formula, 
	\[ 
	d! \sim \sqrt{2 \pi d} \lr{\frac{d}{e  }  }^d  .
	\]
	Thus
	\[ 
	B(d) \sim  
	\frac{1}{\sqrt{2\pi(  1 + \beta )d }}
	\lr{\frac{d}{e \beta}  }^d     \lr{\frac{d}{e  }  }^{-d }
	=
	\frac{\alpha^d}{\sqrt{2\pi(  1 + \beta )d }}	
	=
	\Theta  \lr{\frac{\alpha^d}{\sqrt d}}.
	\]
	In particular, when $ d=2 $, we have
	\[ 
	C_2 \leq B(2) = \frac{2 \cdot 1^2 + 1 \cdot 2^2}{2!}  = 3.
	 \]

\end{appendix}
\end{document}